\documentclass[a4paper,twocolumn,11pt,accepted=2023-10-09]{quantumarticle}
\pdfoutput=1
\usepackage{hyperref} % for cross-referencing of chapters
\usepackage{graphicx} % multiple figures
\usepackage{amsthm} % for definitions and theorems
\usepackage{amsmath}
\usepackage{bm}
\usepackage{mathtools}
\usepackage{bbold}  % for the bold style of identity matrices etc...
\usepackage{amsfonts} % for Lie algebra symbols g,h,etc...
\usepackage{physics} % for bra-ket notation
\usepackage{graphicx}
\usepackage{xcolor} % for text coloring
\usepackage{centernot}
\usepackage[T1]{fontenc}
\usepackage{hyperref}
\usepackage{braket}
\usepackage{siunitx}

\theoremstyle{definition}

\theoremstyle{plain}

\newtheorem{lemma}{Lemma}[section]

\begin{document}

\title{On tests of the quantum nature of gravitational interactions in presence of non-linear corrections to quantum mechanics}

\author{Giovanni Spaventa}
\email{giovanni.spaventa@uni-ulm.de}
\affiliation{Institute of Theoretical Physics and IQST, Universität Ulm, Albert-Einstein-Allee 11 D-89081, Ulm, Germany}

\author{Ludovico Lami}
\email{ludovico.lami@gmail.com}
\affiliation{Institute of Theoretical Physics and IQST, Universität Ulm, Albert-Einstein-Allee 11 D-89081, Ulm, Germany}
\affiliation{QuSoft, Science Park 123, 1098 XG Amsterdam, the Netherlands}
\affiliation{Korteweg--de Vries Institute for Mathematics, University of Amsterdam, Science Park 105-107, 1098 XG Amsterdam, the Netherlands}
\affiliation{Institute for Theoretical Physics, University of Amsterdam, Science Park 904, 1098 XH Amsterdam, the Netherlands}

\author{Martin B. Plenio}
\email{martin.plenio@uni-ulm.de}
\affiliation{Institute of Theoretical Physics and IQST, Universität Ulm, Albert-Einstein-Allee 11 D-89081, Ulm, Germany}

\maketitle

\begin{abstract}
When two particles interact primarily through gravity and follow the laws of quantum mechanics, the generation of entanglement is considered a hallmark of the quantum nature of the gravitational interaction. However, we demonstrate that entanglement dynamics can also occur in the presence of a weak quantum interaction and non-linear corrections to local quantum mechanics, even if the gravitational interaction is classical or absent at short distances. This highlights the importance of going beyond entanglement detection to conclusively test the quantum character of gravity, and it requires a thorough examination of the strength of other quantum forces and potential non-linear corrections to quantum mechanics in the realm of large masses.
\end{abstract}

{\em Introduction ---} Does the gravitational field require quantization and if
so, how do we formulate the correct theory of quantum gravity? Despite intense 
research and debate, these key questions at the interface of quantum mechanics 
and general relativity, two theories that have revolutionised physics in the early
parts of the 20${}^{th}$ century, remain open to this day. These issues have been
hotly discussed already at the 1957 Chapel Hill Conference on the Role of Gravitation 
in Physics~\cite{Feynman57}. At the time not only did Richard Feynman state that ``we're in trouble if 
we believe in quantum mechanics but {\em don't} quantize gravitational theory''; he also %but 
supported this assessment with a \emph{Gedankenexperiment}. He considered a massive 
particle to be placed in a coherent superposition of its spatial degrees of freedom 
by first preparing an internal spin degree-of-freedom of this particle in coherent 
superposition and allowing this particle to pass through a Stern-Gerlach apparatus. 
Then, crucially, he assumed this particle to interact purely gravitationally with yet 
another massive particle.
Indeed, treating the gravitational interaction as either classical or quantum mechanical
results in very different quantum states of the two particles and thus experimental results
of subsequent measurements. The particles would emerge with correlated positions that are 
described either by an incoherent mixture or by a coherent superposition, respectively. The latter,
in modern quantum information parlance referred to as an entangled state, let Feynman to 
conclude that ``We would then have to analyze through the channel
provided by the gravitational field itself via the quantum mechanical amplitudes. Therefore,
there must be an amplitude for the gravitational field''. 

Hence, as discussed further in~\cite{LindnerP2005,KafriT13,KafriTM14,KrisnandaTP2017} the 
observation of the generation of entanglement between two massive particles
would certify the quantum mechanical character of their gravitational interaction.
If this interaction is assumed to be mediated by a gravitational field and the local 
dynamics follows the laws of quantum mechanics, then this would falsify
the assumption of a classical force carrier and thereby imply the non-classical 
nature of the gravitational field \cite{MarlettoV17,MarshmanMB20,GalleyGS21,ChristodoulouBA+22,Carney22}.
Note that in this argument the existence of a gravitational field is taken 
as an assumption that is not directly verified by the experiment \cite{FragkosKP22}.

Nevertheless, one must exercise caution when embarking upon an argument of this 
nature based on experimental data as the strength of the conclusion is inherently 
tied to the validation of the assumptions that underlie the very foundation of the
logical inference.
%However, care needs to be taken when pursuing this type of argument based on experimental data because the conclusion is only as strong as the verification of the validity of the assumptions on which the logical inference is built. 
Indeed, recently models have been proposed in which, even when starting from 
product states, entanglement may be generated with semiclassical models of 
gravity. These models include hybrid quantum-classical ensembles~\cite{HallR18} 
and variations of Bohmian quantum mechanics~\cite{DonerG22}. Both models 
violate fundamental assumptions that enter the proofs that a classical 
force carrier cannot create entanglement and their validity would have to
be established experimentally.   

Here we add another facet to this discussion with models where quantum mechanics 
exhibits weak non-linear corrections to the local\footnote{By local we mean 
that the non-linearity is added as a perturbation to the \textit{free dynamics} of 
each particle, as opposed to adding new interaction terms.} dynamics of massive particles. 
Non-linearities may, for example, emerge in models that couple classical 
gravity to quantum dynamics~\cite{Kibble78,KibbleR80} --- the alternative model 
that we wish to exclude experimentally with the observation of entanglement. 
In the absence of any quantum mechanical interaction the models that we will
consider map product states to product states. However, in any realistic experiment 
there will be residual, albeit very weak, quantum interactions (e.g., electromagnetic 
forces --- Casimir forces, dipolar interactions, etc.)~\cite{PedernalesMP+20}. 
As we will proceed to show, under such models the entangling effect of these quantum 
forces -- even when they are much weaker in strength than gravity at the same distance -- 
can be amplified by non-linear corrections to quantum mechanics. This, in turn, can lead to an entanglement evolution closely resembling that expected
under the influence of quantum mechanical gravity.

As a result, the observation of entanglement can be held to witness the quantum
mechanical nature of the gravitational interaction without loopholes only if 
additional measurements limit the strength of both, non-gravitational forces
{\em and} possible non-linear corrections to the local quantum dynamics at the
length, time and mass scales in question. We discuss possible tests of such non-linear 
corrections to quantum mechanics. 

Non-linear corrections to quantum mechanics may, at first sight, appear a 
daring assumption for a variety of reasons. First, it has been demonstrated 
that specific examples of non-linear quantum mechanics that are, for example,
consistent with the framework put forward by Weinberg~\cite{Weinberg89,Weinberg89a},
allow for signalling~\cite{Gisin89,Czachor91,Polchinski91}, an observation 
that would place considerable strain on notions of causality in physics. It
should be pointed out, however, that proposals have been put forward that 
enable non-linear quantum evolutions to become consistent with Minkowski 
causality~\cite{Kent05,rembielinski2020nonlinear}. Secondly, one may argue that in atomic physics 
the presence of non-linear corrections has been the subject of stringent 
experimental tests that have resulted in tight bounds on possible non-linearities 
in specific models~\cite{BollingerHI+89}. However, such tests have been
carried out only on the spin degrees of freedom of the constituents of atomic 
nuclei.
Under reasonable experimental parameters, tests of the quantum
character of the gravitational interaction require particles
that contain of order $10^{12}$ or more nucleons. This puts them 
significantly beyond the mass range where exhaustive examinations 
of non-linear corrections have been conducted up to this point.
%and are therefore far removed from the mass range in which stringent 
%test of non-linear corrections have been carried out so far. 
It is perfectly conceivable that non-linear corrections to quantum mechanics
scale with mass or energy of the test-masses. This could lead to considerable 
enhancements of non-linear effects for a test-mass with $10^{12}$ or more 
nucleons over those experienced in experiments with single atoms that have
been tested so far~\cite{BollingerHI+89}. 
As a result, the existence of such non-linearities for massive particles 
retains the status of a plausible hypothesis, especially since, as we stressed earlier, such non-linear corrections 
to quantum mechanics appear in
theoretical models that couple classical gravity to 
quantum mechanical matter~\cite{Kibble78,KibbleR80}, the very alternative that 
the observation of entanglement generation
is aiming to exclude.

This motivates the following examination of the consequences of non-linear corrections 
to quantum mechanics in the context of experimental tests of the quantum character of 
gravity in interferometry with massive particles. This work represents a proof of principle, i.e., we construct a particular model where nonlinearities boost entanglement generation, while a general analysis of the effect is postponed to future work.\\

{\em Outline --} The paper is organised as follows: in the first section we introduce Weinberg's framework for non-linear extensions to quantum mechanics. We then proceed with a brief description of typical experimental set-ups in tests of the quantum character of gravitational interactions, and introduce a specific class of non-linear models that are used throughout the paper. After computing the entanglement dynamics and %discuss 
considering the implications of these results for experimental tests, we conclude by discussing how to detect non-linearities in these models via frequency measurements.\\

{\em Weinberg's framework --} There is no unique manner in which to incorporate non-linear terms into quantum mechanics. In this section we introduce an elegant (and rather general) framework due to Weinberg that encompasses a wide variety of possible non-linear extensions to quantum mechanics~\cite{Weinberg89,Weinberg89a}. It was introduced with the express purpose in mind of providing non-linear extensions of quantum mechanics that can then be subjected to experimental test~\cite{BollingerHI+89}.

In standard quantum mechanics, we associate with every physical system
a Hilbert space $\mathcal{H}$, whose elements $\ket{\psi}$ represent its states. On the other hand, physical observables are associated with Hermitian operators $A:\mathcal{H}\to\mathcal{H}$ or, equivalently, with bilinear functions $a:\mathcal{H}\times\mathcal{H}^*\to \mathbb{R}$ such that $a(\ket{\psi},\bra{\psi})=\braket{\psi|A|\psi}$. The time evolution is then generated by the Hamiltonian $H$ via the Schr\"odinger equation $\frac{d}{dt}\ket{\psi}=-i H\ket{\psi}$. Furthermore, there exists a simple prescription for composing (non-interacting) subsystems $A,B$ with Hamiltonians $H_A,H_B$. In particular, the Hilbert space of $A+B$ is constructed as the tensor product $\mathcal{H}=\mathcal{H}_A\otimes \mathcal{H}_B$ of the Hilbert spaces of the two subsystems, and the Hamiltonian $H_{AB}=H_A\otimes\mathbb{1}+\mathbb{1}\otimes H_B$ correctly induces a time evolution on the composite system $A+B$ which is just the free evolution of $A$ under $H_A$ and the free evolution of $B$ under $H_B$, as one expects. Note that, by  
introducing local bases $\{\ket{k}\}_{k=1}^N$ and $\{\ket{j}\}_{j=1}^N$ for $A$ and $B$ respectively, any bipartite state $\ket{\Psi}$ of the composite system can be written as $\ket{\Psi}=\sum_j \ket{\psi_j}\otimes\ket{j}=\sum_k \ket{k}\otimes\ket{\phi_k}$, uniquely defining the states $\ket{\psi_j}$ and $\ket{\phi_k}$. This means that the Hamiltonian composition rule can be equivalently phrased in terms of its associated bilinear function as
\begin{equation}
\begin{split}
      & H_{AB}(\ket{\Psi},\bra{\Psi})=\bra{\Psi}H_A\otimes\mathbb{1}+\mathbb{1}\otimes H_B\ket{\Psi} \\
      & =\sum_j\bra{\psi_j}H_A\ket{\psi_j}+\sum_k \bra{\phi_k}H_B\ket{\phi_k}\,.
\end{split}
\end{equation}
Weinberg's theory, on the other hand, preserves the linear structure of the state space while relaxing one of the main assumption of standard quantum mechanics, i.e., here the existence of Hermitian linear operators representing physical observables. The states of a physical system are still represented by vectors $\ket{\psi}$ in a Hilbert space $\mathcal{H}$, but the bilinear functions $a(\ket{\psi},\bra{\psi})=\bra{\psi} A \ket{\psi}$ associated to physical observables are now replaced by generic (i.e., not necessarily bilinear) functions $a(\ket{\psi},\bra{\psi})$, with the only requirement being their homogeneity of degree one in both entries, i.e., $a(z\ket{\psi},\bra{\psi})=a(\ket{\psi},\bra{\psi}z)=za(\ket{\psi},\bra{\psi})$ $\forall z\in\mathbb{C}$. The reason behind this assumption is related to the requirement that the state $z|\psi\rangle$ be physically equivalent to $|\psi\rangle$ for any complex number $z$. As per the time evolution, it is generated by the Hamiltonian function $h(\ket{\psi},\bra{\psi})$ via the generalized Schr\"odinger equation
\begin{equation}
    \frac{d}{dt}\ket{\psi}=-i\frac{\partial h}{\partial \bra{\psi}}\,,
\end{equation}
which trivially reduces to the usual Schr\"odinger equation in the case of a bilinear $h$, i.e., when $h=\braket{\psi|H|\psi}$ for some Hermitian operator $H$. The framework also generalises the usual prescription for composing subsystems $A,B$ with Hamiltonian functions $h_A(\ket{\psi},\bra{\psi}),h_B(\ket{\phi},\bra{\phi})$. As in the previous case, we introduce local bases $\{\ket{k}\}_{k=1}^N$ and $\{\ket{j}\}_{j=1}^N$ for $A$ and $B$ respectively, which again uniquely define the states $\ket{\psi_j}$ and $\ket{\phi_k}$ through the decompositions of an arbitrary composite state $\ket{\Psi}$. Finally, we can extend the Hamiltonian functions defined on $A$ and $B$ to the composite system $A+B$ via
\begin{equation}
\begin{split}
    & h_{AB}(\ket{\Psi},\bra{\Psi})=\\ &\sum_j h_A(\ket{\psi_j},\bra{\psi_j})+\sum_k h_B(\ket{\phi_k},\bra{\phi_k})     
\end{split}
\end{equation}
This prescription ensures that each of the two subsystems will evolve under the action of its free Hamiltonian alone, which is what we expect from two non-interacting systems. Furthermore, it guarantees that the dynamics generated by $h_{AB}$ map product states to product states. In other words, when two non-interacting systems are initialized in a product state, their non-linear local evolution will always coincide with each subsystem's free dynamics, and the global dynamics will never entangle them spontaneously. If we consider bilinear Hamiltonian functions of $A$ and $B$, i.e., $h_{A,B}=\bra{\psi}H_{A,B}\ket{\psi}$ for some Hermitian operators $H_A,H_B$, the prescription above reduces to the %linear composition of subsystems via direct sum, i.e.,
standard composition rule $H_{AB}=H_A\otimes\mathbb{1}+\mathbb{1}\otimes H_B$.\\

{\em The typical experimental set-up --} A possible test of the quantum character of the gravitational
interaction that probes its capacity to generate entanglement prepares two test-masses at a distance 
$r$ each in a state with a high degree of delocalisation $\Delta x$. 
This can be achieved with the preparation of a Schr{\"o}dinger cat state \cite{SchmoeleDH+16,BoseMM+17,PedernalesMP+20,CarneyMT21,StreltsovPP22,PedernalesSP22}
or of a squeezed state of motion~\cite{KrisnandaTP20,CoscoPP21,WeissRT+21}. In what follows we adopt the 
former as it allows us to reduce the mathematical description in the limit of large delocalisation to that
of a two-dimensional Hilbert space which, in turn, allows for a straightforward application of Weinbergs 
formalism. 

\begin{figure}[h!]
    \centering
    \includegraphics[width=0.49\textwidth]{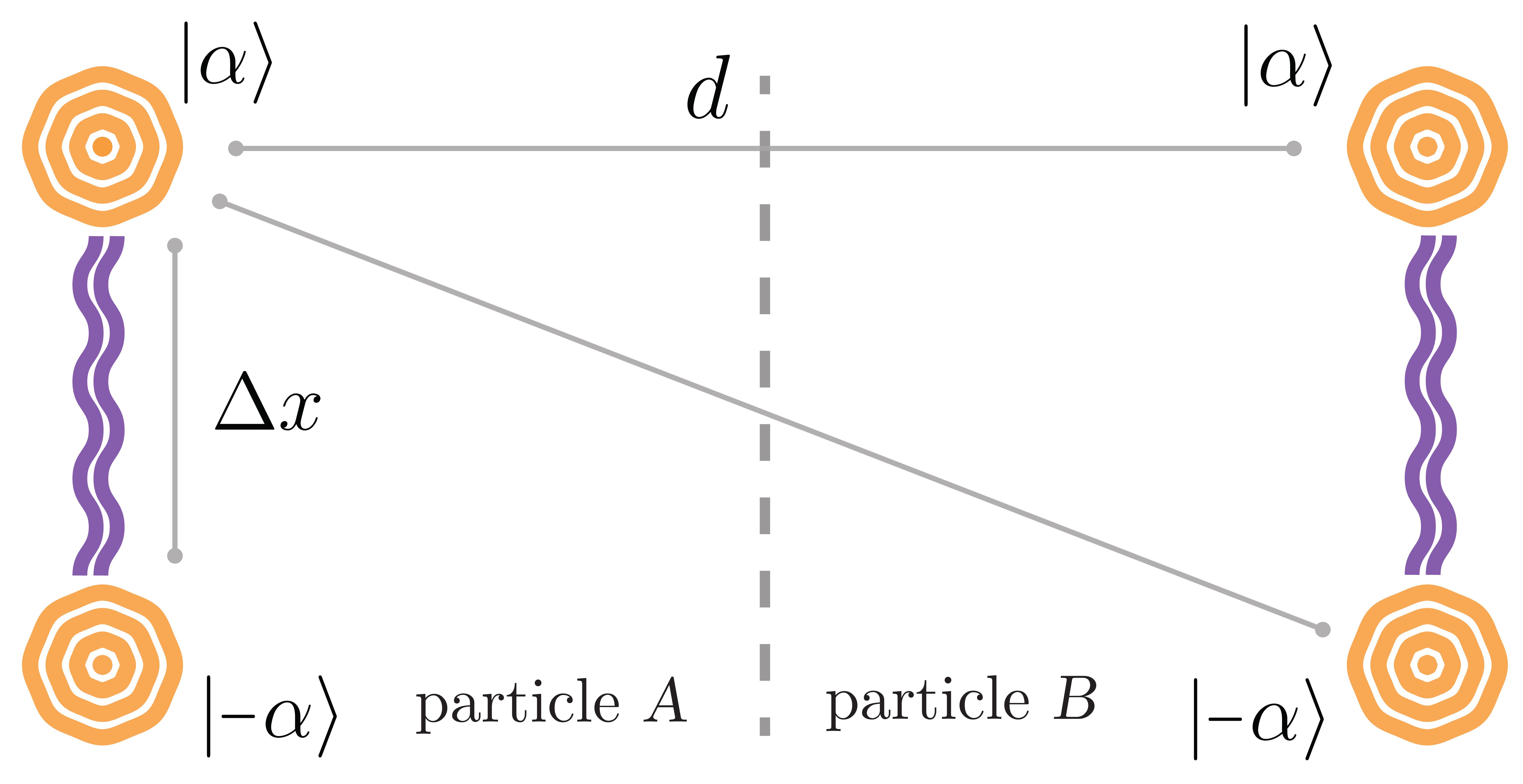}
    \caption{Two particles separated by a distance $d$ are prepared in Schr\"odinger cat states $\mathcal{N}(\ket{\alpha} + \ket{-\alpha})$ with normalisation constant $\mathcal{N}$, such that the spatial separation between the two components $\ket{\alpha}$ and $\ket{-\alpha}$ is given by $\Delta x$. A distance-dependent interaction between the particles will entangle them over time, by accumulating a relative phase between the components $\ket{\pm\alpha}\otimes\ket{\mp\alpha}$ and $\ket{\pm\alpha}\otimes\ket{\pm\alpha}$,.}
    \label{fig:setup}
\end{figure}

Hence, we have in mind the situation depicted in Fig.~\ref{fig:setup} 
where the wave-function of each test-mass is prepared in a Schr\"odinger 
cat state $\mathcal{N}(\ket{\alpha}+\ket{-\alpha})$ with $|\alpha\rangle = e^{-|\alpha|^2/2}\sum_{n=0}^{\infty} \frac{\alpha^n}{\sqrt{n!}}|n\rangle$, normalisation 
constant $\mathcal{N}$ and $|\alpha|$ large enough so
that the overlap between $\ket{\alpha}$ and $\ket{-\alpha}$, given by $|\braket{\alpha|-\alpha}|
= e^{-2|\alpha|^2}$, is negligible. Furthermore, we assume that the 
spatial dynamics of the coherent states is negligible, e.g., because the 
particles are in freefall and the experiment is sufficiently short to 
ensure that the dominant effect is the build-up of a relative phase. 
This allows us to introduce a qubit representation for each test-mass, 
i.e., $\ket{0}\equiv\ket{\alpha}$, $\ket{1}\equiv\ket{-\alpha}$.
% with $\langle 0|1\rangle=0$. 
The initial state is then $\ket{\psi(0)}=(\ket{0}+\ket{1})\otimes(\ket{0}+\ket{1})/2$. Under quantum mechanics, by the symmetry of the set-up and irrespective 
of the precise origin of the force, a distance-dependent interaction 
between the two test-masses will evolve $\ket{\psi(0)}$ into $\ket{\psi(t)}=\ket{00}+e^{ig t}(\ket{01}+\ket{10})+\ket{11}$, i.e., the 
components $\ket{01}$ and $\ket{10}$ will accumulate a phase relative to 
$\ket{00}$ and $\ket{11}$. This dynamics is described by an interaction 
term of the form $gH_{\rm int}=g\left(\ket{01}\bra{01} + \ket{10}\bra{10}\right)$, 
where $g$ determines the coupling strength. We now add a non-linear 
perturbation of the form
\begin{equation}
    h^a_Y(\ket{\psi},\bra{\psi})=\frac{\braket{\psi|\sigma^a_y|\psi}^2}{\braket{\psi|\psi}}\,,
    \label{non-linear}
\end{equation}
to the {\em local} quantum 
dynamics of each subsystem, where $\sigma^a_y=i\ket{0}\bra{1}-i\ket{1}\bra{0}$ acts on subsystem $a$, with $a=1,2$. The two perturbations $h^1_Y,h^2_Y$ can then be composed into a global Hamiltonian function $h_Y$ on the bipartite system $1+2$, as mandated
in Weinberg's approach~\cite{Weinberg89,Weinberg89a}. We stress that this prescription ensures that the dynamics due 
to $h_Y$ maps product states to product states. Then the total Hamiltonian 
function reads
\begin{equation}
    h(\ket{\psi},\bra{\psi})= g\braket{\psi|H_{\rm int}|\psi} + \epsilon\, h_Y(\ket{\psi},\bra{\psi})\,.
\end{equation}
It is important to stress that, in order to obtain the numerical results presented in this paper no explicit composition of the Hamiltonian functions $h_Y^1,h_Y^2$ has been computed. 
%Indeed, the composed function $h_Y$ is defined as precisely the one that induces the correct free evolution on the two subsystems, and this can be exploited as follows. 
Instead, we are making use of the Trotter decomposition and choose a time step $dt$ small enough, such that we can simulate the global evolution from $t$ to $t+dt$ via the factorization of the global dynamics as a composition of: a) free nonlinear evolution of A under its associated Hamiltonian function $h_Y^1$, b) free nonlinear evolution of B under its associated Hamiltonian function $h_Y^2$, and c) linear evolution of $A$ and $B$ under the interaction term $H_{\rm int}$. Thus, in the numerical simulation we only use local nonlinear evolutions and we can avoid computing $h_Y$ explicitly.\\
Naturally, the linear part of the Hamiltonian function, $H_{\rm int}$, alone will 
lead to an oscillatory entanglement dynamics with angular frequency $g$. In
what follows we will demonstrate, perhaps surprisingly, that the non-linear 
contribution to the local Hamiltonian, $\epsilon h_Y$, can enhance the
rate of oscillation of the entanglement between the two test-masses as compared
to that obtained under standard quantum mechanics. 

As we are dealing with noise-free pure state dynamics, we quantify the entanglement 
of the evolved state $\ket{\psi(t)}$ by the von Neumann entropy of the reduced density 
matrix $\rho_A(t) =\text{tr}_B[\ket{\psi(t)}\bra{\psi(t)}]$ of one of the test-masses, 
i.e., $E(\ket{\psi(t)})=-\text{tr}[\rho_A(t) \log(\rho_A(t))]$~\cite{Bennett-distillation, PlenioV07}. 
Let us denote by $E_{(g,\epsilon)}(t)$ the entanglement dynamics obtained 
for a coupling strength $g$ of the linear part  $H_{\rm int}$ and a strength 
$\epsilon$ of the non-linear correction $h_Y$. We chose the entanglement 
dynamics $E_{(1,0)}(t)$ corresponding to linear quantum mechanics with 
unit coupling strength $g$ as the reference. As one observes from the inset
of Fig.~\ref{fig:epsofg}, with increasing $\epsilon$ the non-linear correction 
$\epsilon h_Y$ is able to boost the oscillation frequency rate of the entanglement 
dynamics between the two test-masses while maintaining the same shape with maximal contrast. 
It is then natural to ask whether for any $g<1$ there is a 
choice $\epsilon = \epsilon^*(g)$ such that the frequency of oscillation of 
entanglement of $E_{(g,\epsilon^*(g))}$ matches that of our reference $E_{(1,0)}(t)$. 
As it proves challenging to obtain analytical expressions for $\epsilon^*(g)$, 
Fig.~\ref{fig:epsofg} shows the numerically determined relationship which 
is excellently fitted by
\begin{equation}
    \epsilon^{*}(g) = \frac{1-g}{a_1 g - a_2 g^2} = \frac{1-g}{4.587 g - 4.299 g^2}\,,
    \label{epsilong}
\end{equation}
an invertible function with inverse $g^*(\epsilon).$

\begin{figure}[h!]
    \centering
    \includegraphics[width=0.5\textwidth]{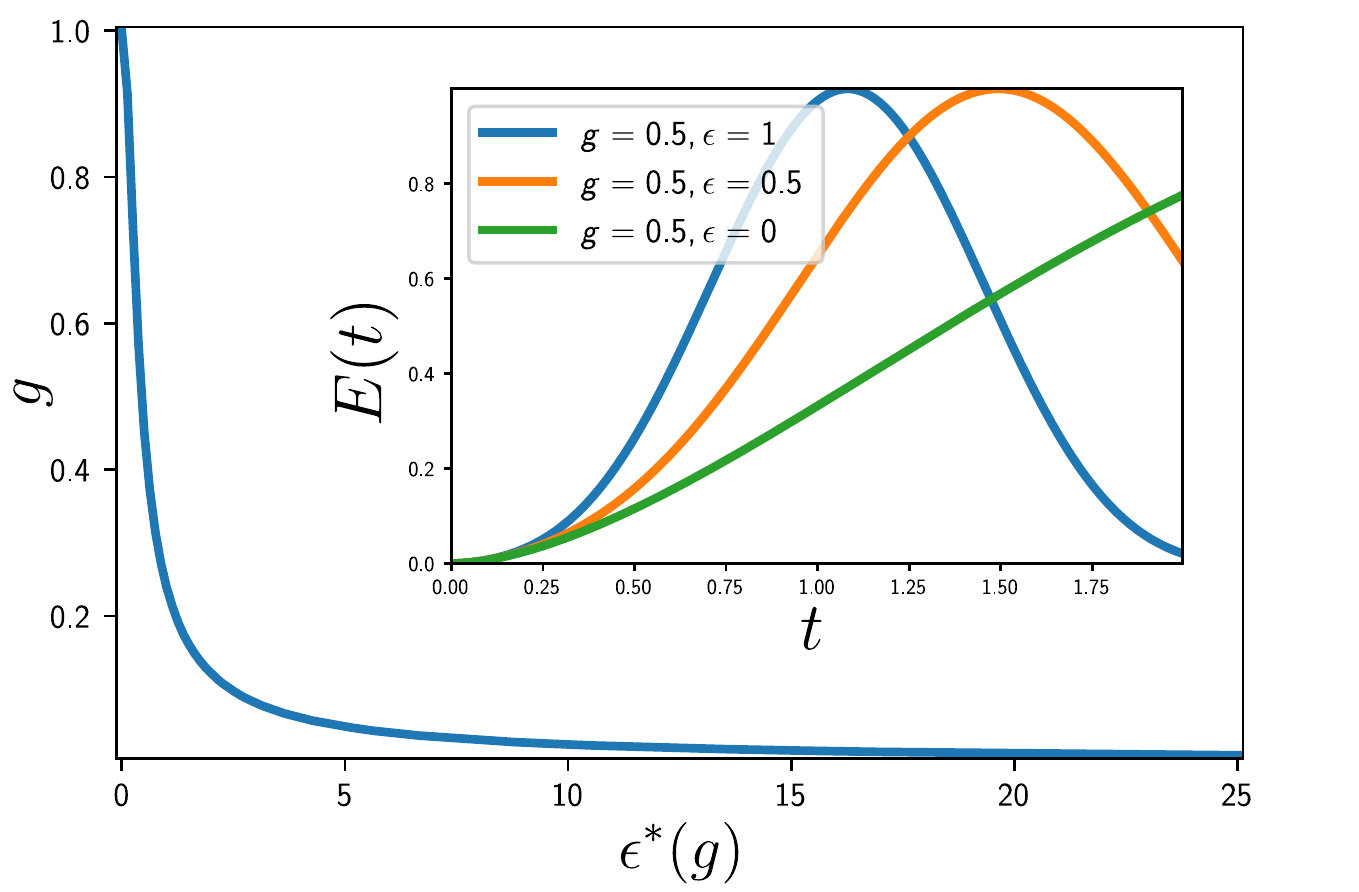}
    \caption{Relation between $g$ and the value $\epsilon^*$ that maximises the overlap with the curve corresponding to $g=1,\epsilon=0$. Equivalently, given a certain non-linear strength $\epsilon$, the value $g^*(\epsilon)$ is the strength of the residual interaction that would reproduce the entanglement dynamics associated with a linear model with $g=1$. (Inset): Entanglement of the composite state as a function of time for three different values of $\epsilon$. The non-linear effects are responsible for boosting the rate at which entanglement is generated.}
    \label{fig:epsofg}
\end{figure}

This relationship shows that, in principle, the rate of oscillation of the 
entanglement can be amplified arbitrarily by the presence of the non-linear 
correction to the local dynamics while maintaining a functional profile closely
following that of the linear case $E_{(1,0)}(t)$. This statement can be made
more quantitative by computing how well the non-linear entanglement dynamics
$E_{(g,\epsilon^*)}(t)$ approximates $E_{(1,0)}(t)$. We measure this with the mean 
squared deviation
\begin{equation}
    d(\epsilon) = \frac{1}{T}\int_0^T dt\, |E_{(g^*(\epsilon),\epsilon)}(t)-E_{(1,0)}(t)|^2\,,
\end{equation}
and plot the results in Fig. \ref{fig:dofg}.

\begin{figure}[h!]
    \centering
    \includegraphics[width=0.5\textwidth]{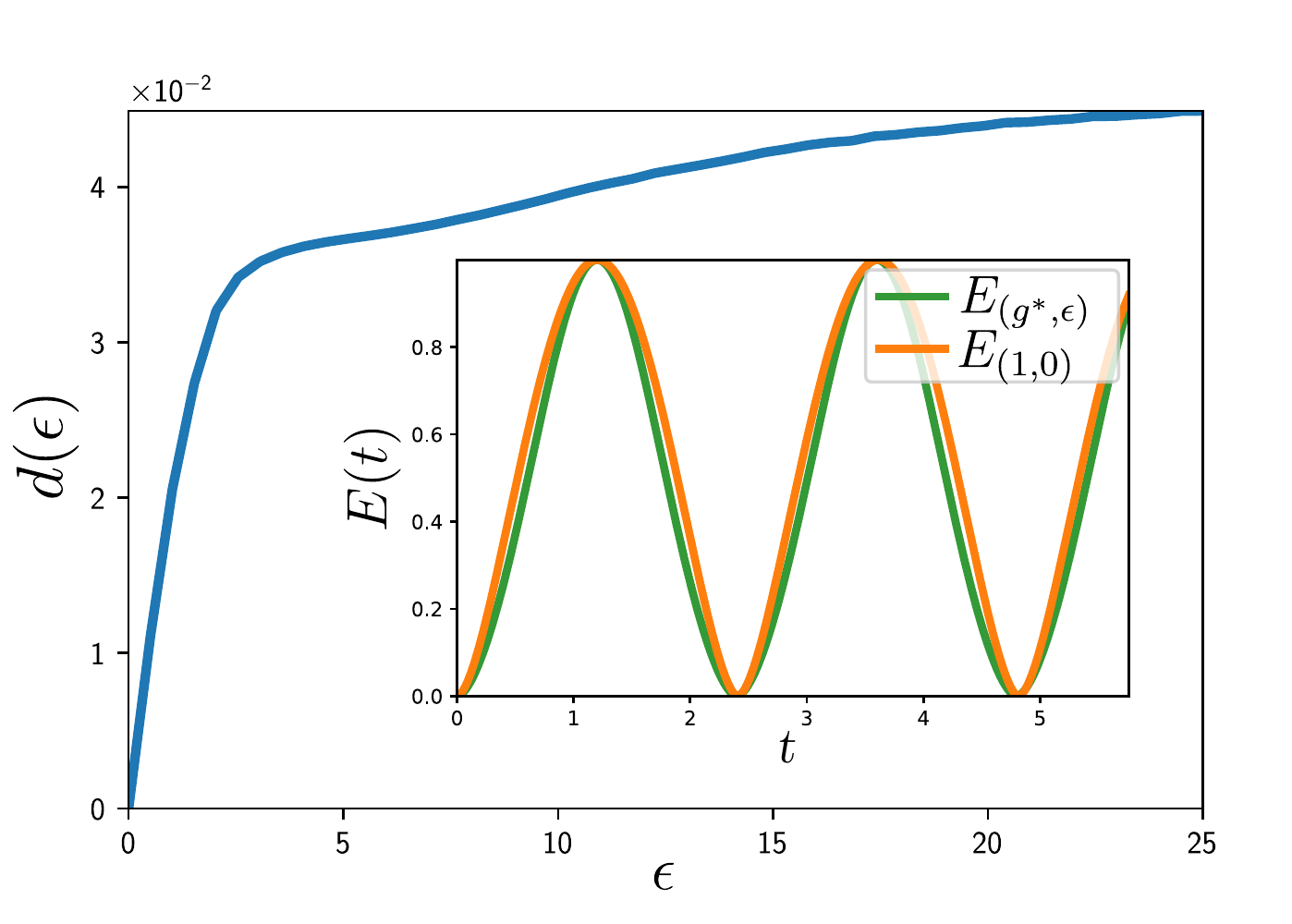}
    \caption{Mean square deviation $d(\epsilon)$ between the entanglement dynamics $E_{(1,0)}(t)$ corresponding to a linear theory with unit coupling, and the one corresponding to a non-linear theory with non-linear strength $\epsilon$ and corresponding coupling $g^*(\epsilon)$. (Inset): An explicit example in which the curve $E_{(1,0)}(t)$ is reproduced by a non-linear theory with $\epsilon=0.3$}
    \label{fig:dofg}
\end{figure}

The expression in Eq.~\eqref{epsilong} allows us to discuss the implications of these findings for realistic experimental parameters. To this end we compare the gravitational force between the two oscillators with the Casimir force between them\footnote{Here we ignore many orders of magnitude stronger interactions resulting from static electric dipole moments of the testmasses~\cite{PedernalesMP+20}.}. 
For two dielectric spherical particles of mass $M_1$ and $M_2$ at distance $r$ the gravitational interaction will 
lead to the accumulation of a relative phase at the rate
\begin{equation}
    \omega_g=\frac{G M_1M_2}{\hbar r}
\end{equation}
while the phase due to the Casimir force will accumulate at a rate \cite{CasimirPolder1948}
\begin{equation}
    \omega_c=\frac{23 c R_1^3 R_2^3}{4\pi r^7}\left(\frac{\epsilon_d - 1}{\epsilon_d + 2}\right)^2
\end{equation}
where $R_1$ and $R_2$ are the radii of the two particles, $\epsilon_d$ is the dielectric constant and $c$ is the speed of light. For definiteness, assume the case of a diamond which has $\epsilon_d=5.7$ with a mass density of $\rho = \SI{3.51e3}{\kg/\m^3}$. If we take $R_1 = R_2 = \SI{1.307}{\um}$, $r=\Delta x= \SI{200}{\um}$, for the arrangement in Fig.~\ref{fig:setup}, we find
$\ket{\psi(t)}=\ket{00}+e^{i\omega_{g/c} t}(\ket{01}+\ket{10})+\ket{11}$ with
$\omega_g=1$ Hz and $\omega_c=\SI{0.096}{Hz}$.
Normally, the latter would be considered negligible compared to the former. 
However, from Eq.~\eqref{epsilong} we find that a non-linear correction $h_Y$ 
with a strength of the order of $\epsilon\sim \SI{2.26}{Hz}$  will be sufficient to 
boost the frequency of oscillation of $E_{(\omega_c,\epsilon)}(t)$ purely due 
to Casimir interaction under a non-linear quantum mechanics to equal that 
of $E_{(\omega_g,0)}(t)$ expected from a quantum mechanical gravity and no 
Casimir force under linear quantum mechanics. \\

{\em Probing non-linear corrections to QM in a harmonic oscillator --}
As we have seen, the presence of non-linear extensions of quantum mechanics 
can have significant effects on the dynamics and hence on the conclusion that can be drawn
from such experiments. It is thus natural to examine whether non-linearities 
of the magnitude discussed above would already lead to observable
consequences in current or soon-realisable experiments on a single massive particle 
subject to a harmonic potential. As the preparation of highly non-classical states is 
very challenging experimentally, we consider a harmonic oscillator initially prepared 
in a thermal state. We displace the trap center suddenly in space such that relative 
to the new coordinates the harmonic oscillator is in a displaced thermal state whose 
mean position and momentum are expected to oscillate periodically in time. It is reasonable 
to expect that the oscillation frequency will be a function of the strength of the non-linear 
corrections to quantum mechanics. Such frequency changes will be detectable more easily for a 
harmonic oscillator that is cooled to the ground state which is what we assume for the following. 
This has the added benefit of simplifying the analytical treatments for important limiting cases.

For the experimental set-up shown in Fig.~\ref{fig:setup} we were able to neglect the 
spatial dynamics and thus assumed the coherent state amplitude $\alpha$ to be time-independent. 
Now, we are interested in the full spatial dynamics of the harmonic oscillator
including the non-linearities for arbitrary $\alpha$. For its mathematical description
we need to define three operators in such a manner that for large $\alpha$ they reduce 
to the Pauli-operators in Eq.~\eqref{non-linear} for the subspace spanned by $|\pm\alpha\rangle$. 
To this end we define the parity operator 
\begin{equation}
    P=e^{i\pi a^\dagger a},
\end{equation}
and the projector $\mathbb{P}_{\alpha}$ onto the subspace spanned by $|\pm\alpha\rangle$ which yields
by lengthy but direct calculation
\begin{equation}
    \| \mathbb{P}_{\alpha} P \mathbb{P}_{\alpha} - \sigma_x\|_1 = 2e^{-2|\alpha|^2}
\end{equation}
where $\|\cdot\|_1$ is the trace norm and $\sigma_x \coloneqq |\alpha\rangle\langle -\alpha| + |-\alpha\rangle\langle \alpha|$. 
Furthermore, using the displacement operator $D(\beta)=e^{\beta a^\dagger - \beta^* a}$, we define 
\begin{equation}
    Y = D^\dagger(\beta) P D(\beta)
    \label{Yop}
\end{equation}
where $\beta\in\mathbb{C}$ is for now a free parameter of the model. The requirement that the operator $Y$ defined above becomes a $\sigma_y$ in the two-dimensional subspace spanned by $\ket{\pm\alpha}$ when $|\alpha|\gg 1$ defines the parameter $\beta$ in terms of $\alpha$ as $\beta = -i\frac{\pi}{8\alpha}$, as for such value of $\beta$ one gets
\begin{equation}
    %|\mathbb{P}_{\alpha}Y\mathbb{P}_{\alpha} - \sigma_y|_1 \le \frac{\pi^2}{16|\alpha|^2}
    \|\mathbb{P}_{\alpha}Y\mathbb{P}_{\alpha} - \sigma_y\|_1 = 2\left(1-e^{-\frac{\pi^2}{32|\alpha|^2}}\right)\,,
\end{equation}
obtained by making use of the relation $D(\alpha+\beta) = D(\alpha)D(\beta)e^{-i \Im(\alpha\beta^*)}$ which
yields $D(-i\frac{\pi}{8\alpha})|\alpha\rangle = e^{i\pi/8}|\alpha - i\frac{\pi}{8\alpha}\rangle$. Thus
for large $\alpha$ the operator $Y$ reduces to the $\sigma_y$-operator in the subspace spanned by 
$|\pm\alpha\rangle$. In order to be as general as possible, we can study this problem by keeping $\beta$ as a free parameter, knowing that this model reduces to the qubit model studied earlier only when $\beta$ satisfies the constraint above.

In keeping with the mathematical description of the experiment in Fig.~\ref{fig:setup} we proceed to 
choose the non-linear correction to quantum mechanics as
\begin{equation}
    h_Y(\ket{\phi},\bra{\phi})=\frac{\braket{\phi|Y|\phi}^2}{\braket{\phi|\phi}}\,.
\end{equation}
The resulting Hamiltonian function is then
\begin{equation}
    h(\ket{\phi},\bra{\phi}) = \omega_0\braket{\phi|a^{\dagger}a|\phi} + \epsilon h_Y(\ket{\phi},\bra{\phi})\,,
\end{equation}
which yields the non-linear Schr\"odinger equation
\begin{equation}
    \frac{d\ket{\phi}}{dt} = -i \Big[ \omega_0 a^\dagger a + \epsilon\Big( 2 \braket{\phi|Y|\phi}Y - \braket{\phi|Y|\phi}^2\mathbb{1} \Big)\Big]\ket{\phi}\,.
    \label{eq:schreq}
\end{equation}
We proceed by solving the resulting dynamics by making use of the formalism of \textit{generalized coherent states} (GCS)~\cite{zhang1990coherent}. These are states whose evolution is described by classical Hamilton equations on a manifold that plays the role of a classical phase space. 
As presented in Appendix~\ref{appendix:GCS}, the GCS for this system are, to first order in $\epsilon$, of the form
\begin{equation}
    \ket{\alpha,\theta} 
    \coloneqq D(\alpha,\theta)\ket{0}\,,
\end{equation}
where 
\begin{equation}
    D(\alpha,\theta)=\Big(\mathbb{1}+\epsilon \Pi(\alpha,\theta)+ o(\epsilon^2)\Big)D(\alpha) 
\end{equation}
is a generalization of the usual displacement operator $D(\alpha)$, parametrized by the complex coordinate 
$\alpha$ and by a set of additional complex coordinates $\theta \coloneqq \{\theta_\mu|\mu\in\mathbb{N}\}$. 
The operator $\Pi(\alpha,\theta)$ can be computed explicitly and its expression is found in Appendix~\ref{appendix:GCS}. 
The states $\ket{\alpha,\theta}$ live on a very large manifold, and they coincide with the standard coherent 
states for $\theta_\mu=0$. 
The dynamics $\ket{\alpha(t),\theta(t)}$ on the manifold is specified by a set of Hamilton equations, as 
outlined in~\cite{zhang1990coherent}. In the limit of large amplitudes, $|\alpha|\gg 1$, we find all 
corrections due to the non-linearity $h_Y$ to be suppressed exponentially in $|\alpha|^2$ (as described 
in Appendix~\ref{appendix:GCS}) and the dynamics is extremely well approximated by a trajectory on the 
complex plane $(\Re(\alpha),\Im(\alpha))$ obeying $\dot{\alpha}(t) = -i\omega_0 \alpha(t)$. Thus an initial
state $\ket{\alpha,0}$ with $|\alpha|\gg 1$ will evolve as $\ket{\alpha(t),0}$ with 
$\alpha(t)=\alpha e^{-i\omega_0 t}$ (i.e., as a linear oscillator), irrespective of the size of the 
non-linear strength $\epsilon$.

\begin{figure}[h!]
    \centering
    \includegraphics[width=0.5\textwidth]{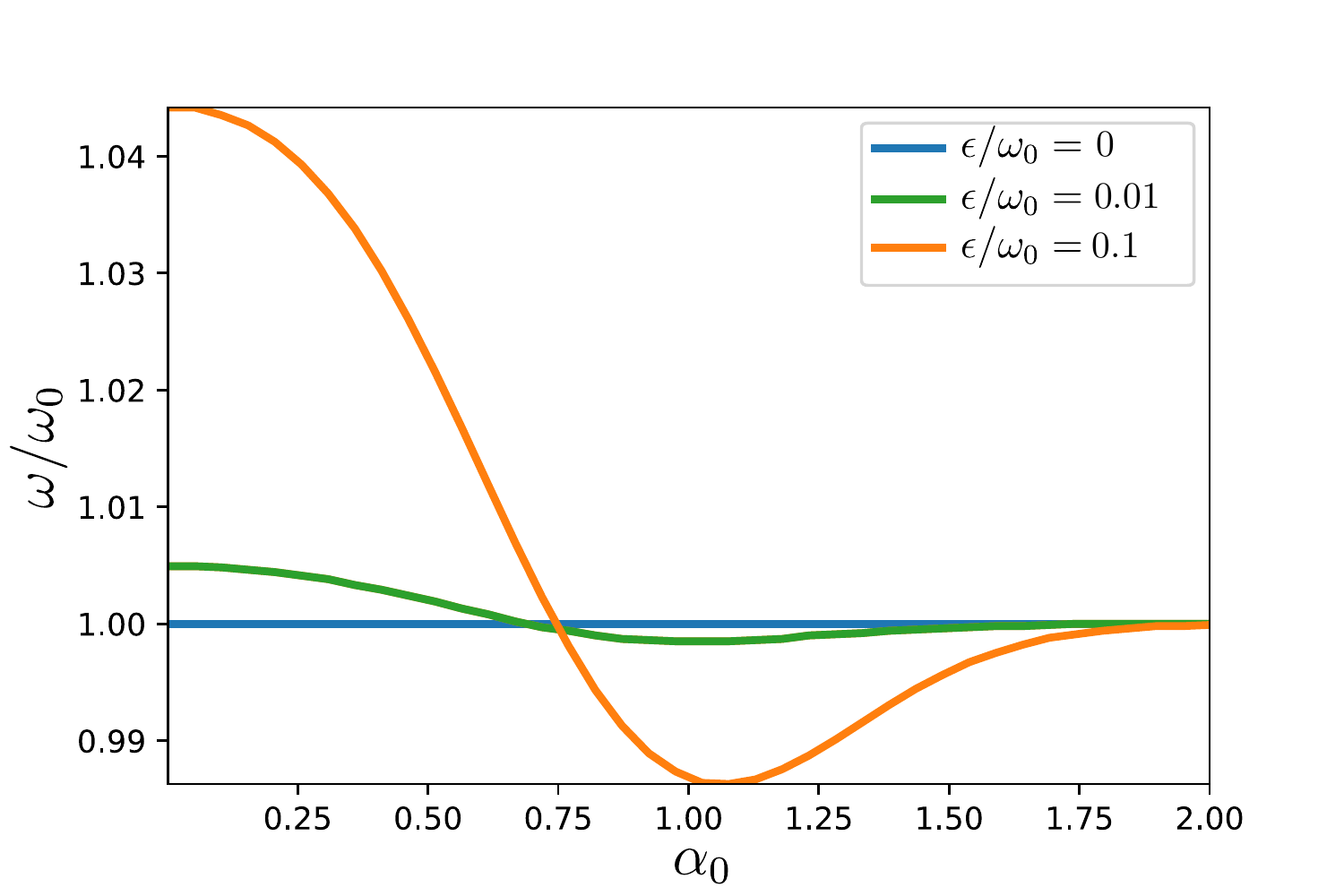}
    \caption{Deviations of the oscillation frequency of the system from the frequency $\omega_0$ of a linear oscillator. The corrected frequency is computed by initializing the system in the exact 
    vacuum (numerically determined) and then displacing it in position by an amount
    $\alpha_0$. The deviation from the frequency of the linear oscillator 
    scales approximately as $\epsilon/(2\omega_0)$ for small amplitudes, while for large amplitudes it is always negative and vanishes exponentially as $-4(\epsilon/\omega_0) e^{-4|\alpha_0|^2}$ (see Appendix~\ref{appendix:GCS} for details).} %Interestingly enough, there is no correction to the frequency for $\alpha_0>2$, regardless of $\epsilon$.
    \label{fig:osc_freq}
\end{figure}

In the regime of small displacement amplitudes, smaller or of the 
order of the spatial extent of the ground state wave function, we 
find in a first order amplitude expansion (see Appendix B) that the 
spatial dynamics still displays a periodic oscillatory behaviour, but 
with a frequency $\omega$ 
that deviates from that of the linear harmonic oscillator, 
$\omega_0$. In order to determine these frequency shifts, one needs to
account for the fact that the unperturbed ground state $\alpha=0$ does not 
approximate well the perturbed ground state (which is the state 
with the lowest energy expectation value). In leading order in the non-linearity 
an approximation to the ground state is given by coherent 
state with amplitude
\begin{equation}
    \alpha_v = 4\beta \frac{\epsilon}{\omega_0}  e^{-4|\beta|^2}\,.
\end{equation}
The construction of GCS for the non-linear oscillator allows for the determination
of the frequency shift (see Appendix~\ref{appendix:GCS}), which for $\beta$ as
large as $|\beta|=\pi/4$ gives 
\begin{equation}
    \frac{\delta\omega_0}{\omega_0} =
    4e^{-\frac{\pi^2}{4}}\left(1-\frac{\pi^2}{4}\right) \frac{\epsilon}{\omega_0}\cong
    \frac{\epsilon}{2\omega_0}\,.
\end{equation}
This is confirmed in Fig.~\ref{fig:osc_freq} where the numerically determined 
exact ground state has been displaced. The frequency shift is largest 
for very small displacement amplitudes $\alpha_0$. In this regime the
accurate determination of the small amplitude of oscillation will however
require a large number of measurements that scales approximately as 
$\alpha_0^{-2}$. Indeed, given a Gaussian state with a width $\sigma$, a 
position measurement will have a resulting variance $\Delta x^2 = \sigma^2$. 
A sequence of $N$ position measurements will reduce the variance by a factor 
of $N$, so that in order to appreciate a displacement $\alpha_0$ it must be 
$\Delta x^2/N=\alpha_0^2$, which implies $N=\sigma^2/\alpha_0^2$. For a thermal 
state with an average number of phonons $n_{\rm th}$, we have $\sigma\sim n_{\rm th}+1$ 
and therefore the number of required measurements will scale as $(n_{\rm th}+1)^2/\alpha_0^2$. 
Thus we need a large $\alpha_0$ in order to minimize the number of required measurements. 
On the other hand the deviation $\omega/\omega_0$ rapidly approaches unity for $\alpha_0>1$  
which rapidly diminishes the benefits of increasing $\alpha_0$. Thus, Fig.~\ref{fig:osc_freq} 
suggests that a displacement of the order of $\alpha_0=1$ provides a good balance between 
size of the frequency shift and amplitude of oscillation. 

As a result the experimental verification of non-linearities of order $\epsilon\sim 2$ Hz 
requires harmonic oscillators of very low
frequency $\omega_0$ whose initial state is 
cooled close to the ground state and displacement
amplitudes of order $\alpha_0\cong 1$. A careful design and analysis of such an experiment beyond 
these simple estimates is interesting but goes
beyond the scope of this work.\\

{\em Discussion --} We have demonstrated that the observation 
of entanglement generation between two massive particles, whose
interaction is dominated by gravity, witnesses its quantum mechanical
character only under additional assumptions that remain to be 
tested experimentally. Notably, the entangling capacity of very
weak quantum forces, much weaker than gravity at the same distance, 
may be amplified by non-linear corrections to local quantum dynamics 
predicted in certain models that couple classical gravity to quantum
matter and may closely mimic expected gravitational dynamics. It is worth mentioning that a complete theoretical analysis of the relation between entanglement and Weinberg nonlinearities is an interesting possibility for future work, as the current paper represents rather a \textit{proof of principle} than a general proof. We
also note that the specific model presented here ignores the effect of
noise that typically accompanies classical gravity coupling to 
quantum matter~\cite{KafriT13,KafriTM14,OppenheimSS+22}, has parameters
that are fine tuned and also suffers in its present formulation from 
the possibility of signalling. 
However, it is certainly possible to construct non-linear extensions to quantum mechanics that do not exhibit signalling. In particular, in ~\cite{Kent05} the author shows that any nonlinear model, where the nonlinearity arises as a dependence on local states, can be emulated by using standard quantum mechanics only, and thus cannot allow signalling. Another example is provided in \cite{rembielinski2020nonlinear} where the authors construct a general class of nonlinear models that do not allow signalling. These models have the property that their corresponding dynamical maps are \textit{convex quasilinear} as opposed to \textit{linear}. These finding suggest that other non-linear models might be found that still display the desired entanglement amplification, while not allowing for superluminal effects.
Nevertheless, this model demonstrates
that a conclusive test of the quantumness of gravitational interaction 
via the observation of entanglement generation
may eventually require, at the very least, additional tests that limit 
non-linear corrections to quantum mechanics. %Deeper
A deeper analysis may 
reveal additional assumptions that may need to be tested independently.

This situation is reminiscent of the long journey towards establishing
conclusive experimental evidence of the non-locality of nature by tests
of Bell inequalities. There, increasingly sophisticated experimental 
tests were devised and realised to close possible loopholes --- that is, 
establish the correctness of the assumption that are underlying the 
logical argument --- until only quantum mechanics and its inherent non-locality 
remained the only reasonable explanation of the observations~\cite{BrunnerCP+14}. 
We expect that experiments to witness quantum
properties of gravity, due to their indirect nature, will experience a similar development. \\

{\em Numerical simulations and figures ---} The data used for Figures~\ref{fig:epsofg},~\ref{fig:dofg}, and~\ref{fig:osc_freq} has been generated via Python codes available at the GitHub repository \url{https://github.com/gspaventa/nonlinear_oscillators} in the form of two Jupyter notebooks.\\

{\em Acknowledgements ---} We thank Julen Pedernales and Kirill Streltsov for critical reading and helpful comments on this manuscript and Susana Huelga for discussions at early stages of this work. This work was supported by the QuantERA projects Lemaqume and 
ExtraQt as well as the Alexander von Humboldt Foundation.

\clearpage
\onecolumngrid
\appendix

\section{Generalized coherent states for the non-linear oscillator}
\label{appendix:GCS}

The generalized coherent states (GCS) \cite{zhang1990coherent} framework provides a recipe for constructing coherent states for a wide class of physical systems. These are states whose evolution can be mapped into trajectories on a symplectic manifold. Furthermore, such trajectories are obtained as solutions of Hamilton equations as in classical mechanics. The simplest case being that of a standard harmonic oscillator with Hamiltonian $H=\omega a^\dagger a$, for which the GCS are the well-known coherent states $\ket{\alpha}$ with $\alpha\in \mathbb{C}$. Indeed, the time evolution of such states can be described as classical orbits $\alpha(t)\in \mathbb{C}$ induced by the classical Hamiltonian $H(\alpha,\alpha^*)=\omega |\alpha|^2$. In addition to being a fascinating topic in mathematical physics, the GCS framework can be applied to various contexts, such as the study of classical limits of quantum theories \cite{yaffe1982large}, or the the analysis of open quantum systems in terms of their parametric representations \cite{calvani2013parametric,Spaventa2022nature}. In this section we construct the GCS system for a single harmonic oscillator with a non-linear perturbation, and we study the impact of these non-linear corrections to its dynamics. Here, we find that even at first order in $\epsilon$ the GCS dynamics for this system is extremely convoluted. However, it is possible to get information about the dynamics from the regimes of very large and very small coherent state amplitudes. In particular, we show that the non-linear oscillator behaves as a linear one for large amplitudes, and we analyze the dynamics around the ground state to find the correction to the bare frequency $\omega_0$. The (non-linear) Hamiltonian that we have to analyze is
\begin{equation}
    H(\phi) := \omega_0 a^\dagger a + \epsilon\Big( 2 \bra{\phi}Y\ket{\phi}Y - \bra{\phi}Y\ket{\phi}^2\mathbb{1} \Big)\,,
\end{equation}
where $Y=D^\dagger(\beta)PD(\beta)$, and $P=e^{-\pi a^\dagger a}$. The expression above defines for any state $\ket{\phi}$ a linear combination of the operators $\{\mathbb{1},a,a^\dagger, a^\dagger a, \epsilon Y\}$. Given that the Heisenberg-Weyl algebra $\mathfrak{h}_4=\text{Span}(\mathbb{1},a,a^\dagger,a^\dagger a)$ of the quantum harmonic oscillator gives rise to the (standard) Glauber coherent states $\ket{\alpha}$, $\alpha\in\mathbb{C}$, we want to see whether it is possible to perturbatively construct the GCS for the non-linear system at hand, and write them as a correction to the unperturbed ones, to first order in $\epsilon$. The natural first step would be to include $\epsilon Y$ as a new generator in $\mathfrak{h}_4$ but, unfortunately, the operators $\{\mathbb{1}\,, a\,, a^\dagger\,, a^\dagger a\,, \epsilon Y \}$ do not form a closed algebra under commutation, since for example
\begin{equation}
    [Y,a]=2Y(\beta+a)\,,\quad [Y,a^\dagger] = 2(\beta -a^\dagger ) Y\,, \quad [Y, a^\dagger a] = 2\beta (a^\dagger Y + Y a)\,.
\end{equation}
This means that additional generators must be added in order to form a closed algebra. To this end, let us define for $j,k\geq0$ the operators
\begin{equation}
 Y_{jk}={a^\dagger}^j Y a^k\,,
\end{equation}
which obey the following commutation relations:
 \begin{equation}
 \begin{split}
    &[Y_{jk},a]=2\beta Y_{jk} + 2Y_{j,k+1}-jY_{j-1,k}\,,\\ &[Y_{jk},a^\dagger]=2\beta Y_{jk}-2Y_{j+1,k}+kY_{j,k-1}\,, \\ & [ Y_{jk},a^\dagger a]= 2\beta (Y_{j+1,k}+Y_{j,k+1})+(k-j)Y_{jk}\,.
    \label{eq:commutators}
\end{split}
\end{equation}   
As per the commutator $[Y_{jk},Y_{rs}]$, it can be written as a combination of operators $N_{jk}={a^\dagger}^j a^k$, thanks to the fact that $Y^2=\mathbb{1}$ and $[a,a^\dagger]=\mathbb{1}$. However, every time such a term is generated, it carries a coefficient $\epsilon^2$ in front of it and is thus negligible if we are interested in a first order treatment. Therefore, even though the algebra would only close with the additional generators $N_{jk}$, we can safely ignore this complication and consider, to first order in $\epsilon$
\begin{equation}
    [Y_{jk},Y_{rs}]\approx 0\,.
\end{equation}

The maximal isotropy subgroup $H$ (with respect to the unperturbed vacuum $\ket{0}$) is obtained by exponentiation of the subalgebra spanned by $\{\mathbb{1},a^\dagger a\}$, as usual. This means that the coset $G/H$ of displacement operators consists of elements of the form
\begin{equation}
    D(\alpha,\theta) = \exp\Big(\alpha a^\dagger - \alpha^* a  + \epsilon\sum_{\mu}(\theta^*_\mu Y_\mu - \theta_\mu Y_\mu^\dagger )\Big)
\end{equation}
where $\alpha$ is a complex number, $\{\theta_\mu\}$ is a (finite) sequence of complex numbers, and $Y_\mu:=Y_{\mu,0}$ are the only operators $Y_{jk}$ that do not annihilate the unperturbed vacuum $\ket{0}$. 
% Using the (right-oriented) Zassenhaus formula \cite{casas2012efficient}, and keeping only those terms that are of first order in $\epsilon$,  we can perform the following expansion
% \begin{equation}
%     e^{A+\epsilon B}=\left[\prod_{n=0}^\infty \exp\left(\epsilon\frac{(-1)^n}{n!}\text{ad}^n_A(B)+o(\epsilon^2)\right)\right] e^A\,,
% \end{equation}
% where 
% \begin{equation}
%     \text{ad}^n_A(B)= \underbrace{[A,[A,[A,...[A,B]]]]}_{n \text{ times}}\,,\quad \text{ad}^0_A(B):=B\,.
% \end{equation}
% Then we can write
% \begin{equation}
%     e^{A+\epsilon B}=\prod_{n=1}^\infty \left( \mathbb{1}+ \epsilon\frac{(-1)^n}{n!}\text{ad}^n_A(B)+o(\epsilon^2) \right)e^A = \left( \mathbb{1}+ \epsilon \sum_{n=0}^{\infty}\frac{(-1)^n}{n!}\text{ad}^n_A(B)+o(\epsilon^2) \right)e^A\,,
% \end{equation}
% Now, we use two properties of the adjoint actions $\text{Ad}_g$, $\text{ad}_A$ of a Lie group and its corresponding Lie algebra, namely that $\text{Ad}_{\exp(A)}= \exp\left(\text{ad}_A\right)$, and $\text{Ad}_g(A)=gAg^{-1}$. This allows us to write
% \begin{equation}
%     e^{A+\epsilon B} = \Big(\mathbb{1} + \epsilon e^{-\text{ad}_A}B +o(\epsilon^2)\Big)e^A = \Big(\mathbb{1} + \epsilon \text{Ad}_{\exp(-A)}B+o(\epsilon^2)\Big) e^A = \Big(\mathbb{1} + \epsilon\, e^{-A} B e^{A}+o(\epsilon^2)\Big) e^A
% \end{equation}
An expansion in powers of $\epsilon$ to first order yields
\begin{equation}
    D(\alpha,\theta)=D(\alpha) + \epsilon\, \int_0^1 d\tau D(\tau \alpha) \sum_{\mu}(\theta^*_\mu Y_\mu - \theta_\mu Y_\mu^\dagger ) D((1-\tau)\alpha)+o(\epsilon^2)\,,
\end{equation}
which, by defining
\begin{equation}
    \Pi(\alpha,\theta):=\sum_\mu  \int_0^1 d\tau D(\tau \alpha) (\theta_\mu^*Y_\mu - \theta_\mu Y_\mu^\dagger) D^\dagger (\tau \alpha)=-\Pi^\dagger (\alpha,\theta)\,,
\end{equation}
can be rewritten as
\begin{equation}
    D(\alpha,\theta)=\Big(\mathbb{1} + \epsilon\, \Pi(\alpha.\theta)\Big) D(\alpha)\,.
\end{equation}
The expression above tells us that, to first order in $\epsilon$, the GCS of the non-linear oscillator can always be constructed as
\begin{equation}
    \ket{\alpha,\theta} = \Big(\mathbb{1} + \epsilon\, \Pi(\alpha.\theta)\Big)\ket{\alpha} + o(\epsilon^2)\,.
\end{equation}

The Hamiltonian function, evaluated on the coherent states $\ket{\alpha,\theta}$, has the expression
\begin{equation}
    H(\alpha,\theta) = \omega_0|\alpha|^2+\epsilon\bra{\alpha} Y \ket{\alpha}^2 + \epsilon \bra{\alpha} \big[a^\dagger a, \Pi\big]\ket{\alpha} + o(\epsilon^2)\,.
\end{equation}
Using the fact that $\bra{\alpha} Y \ket{\alpha}=e^{-2|\alpha+\beta|^2}$, together with the definition
\begin{equation}
    \Omega(\alpha.\theta):= \bra{\alpha} \big[a^\dagger a, \Pi\big]\ket{\alpha}= 2 \Re\Big( \alpha^* \bra{\alpha} a \Pi \ket{\alpha}\Big)\,,
\end{equation}
the equations of motion will have the form
\begin{equation}
\begin{split}
    & \dot{\alpha} = -i \Big( \omega_0\alpha - 4\epsilon(\alpha+\beta)e^{-4|\alpha+\beta|^2}+ \epsilon\frac{\partial \Omega(\alpha,\theta) }{\partial \alpha^*}  \Big)+o(\epsilon^2)\,,\\ &\dot{\theta}_\mu = -i \epsilon\frac{\partial \Omega(\alpha,\theta) }{\partial \theta^*_\mu}+o(\epsilon^2)\,.   
\end{split}
\label{eq:hamilton_eqs}
\end{equation}
Therefore, we need to compute $\Omega$ and its derivatives. First, let us consider the resolution of the identity 
\begin{equation}
    \mathbb{1}=\int d\mu(\alpha) \ket{\alpha}\bra{\alpha}\,,
    \label{eq:res_identity}
\end{equation}
where the integration measure is given by
\begin{equation}
    d\mu(\alpha)= \frac{d\alpha \wedge d\alpha^*}{2\pi i} e^{-|\alpha|^2} = \frac{d x \wedge dy }{\pi} e^{-(x^2+y^2) }\,.
\end{equation}
By making use of Eq.\ref{eq:res_identity} we can rewrite $\Omega(\alpha,\theta)$ as
\begin{equation}
    \Omega(\alpha,\theta) = \bra{\alpha}\big[a^\dagger a, \Pi\big]\ket{\alpha} = \int d\mu(\alpha_1) \Big(\alpha_1\alpha^* \bra{\alpha_1}\Pi\ket{\alpha} + {\alpha_1}^*\alpha \bra{\alpha_1}\Pi\ket{\alpha}^* \Big) = 2 \Re \Big( \int d\mu(\alpha_1)\, \alpha_1\alpha^* \bra{\alpha_1}\Pi\ket{\alpha} \Big)\,,
\end{equation}
and since
\begin{equation}
     \bra{\alpha_1}\Pi\ket{\alpha_2} = e^{-2|\beta|^2} e^{2\beta (\alpha_2-{\alpha_1}^*)} e^{-\frac{1}{2}(|\alpha_1|^2+|\alpha_2|^2+2{\alpha_1}^*\alpha_2)}\sum_\nu \Big[ \theta_\nu^* {\alpha_1^*}^\nu - \theta_\nu \alpha_2^\nu \Big]\,,
\end{equation}
we get
\begin{equation}
\begin{split}
     &\Omega(\alpha,\theta) = 2e^{-2|\beta|^2} \sum_\nu \Re \left[ \alpha^* e^{-2\beta\alpha} e^{-\frac{1}{2}|\alpha|^2 }  \int d\mu(\alpha_1) \left( \theta_\nu^* \alpha_1{\alpha_1^*}^\nu - \theta_\nu\alpha^\nu \right) e^{-2\beta\alpha_1^*} e^{-\frac{1}{2}(|\alpha_1|^2 + 2\alpha_1^*\alpha)}  \right]\\ 
    & =  2e^{-2|\beta|^2} \sum_\nu \Re \left[ \alpha^* e^{-2\beta\alpha} e^{-\frac{1}{2}|\alpha|^2 }  \Big( \theta_\nu^* I_\nu(\alpha) - \theta_\nu\alpha^\nu I_0(\alpha) \Big) \right]\,,
\end{split}
\end{equation}
where we have defined
\begin{equation}
    I_\nu(\alpha) = \int d\mu(\alpha_1) \alpha_1{\alpha_1^*}^\nu e^{-2\beta\alpha_1^*} e^{-\frac{1}{2}(|\alpha_1|^2 + 2\alpha_1^*\alpha)}\,.
    \label{eq:integral_nu}
\end{equation}
The family of integrals above turn out to be exactly solvable, by deploying techniques of complex contour integration and the Cauchy's residue theorem, as described in the following

\begin{lemma}
Given $\nu\in\mathbb{Z}$ and $b,c\in\mathbb{C}$, if $\Re(b)>0$, it is
    $$ I_\nu (b,c)=\int_\mathbb{C} d\mu(\alpha)\,{\alpha^*}^\nu  \alpha\, e^{-\alpha^*( b\alpha - c )} = \frac{1}{(b+1)^2}\Big( \delta_{\nu,0}\, c + \delta_{\nu,1} \Big) .$$
\end{lemma}
\begin{proof}
By making use of the definition of the measure $d\mu(\alpha)$ we can turn $I_\nu(b,c)$ from an integral over the complex plane to an integral over $\mathbb{R}^2$ as following:
\begin{equation}
        I_\nu(b,c) = \iint_{\mathbb{C}} \frac{d\alpha\wedge d\alpha^*}{2\pi i} e^{-(b+1)|\alpha|^2} {\alpha^*}^\nu  \alpha\, e^{c \alpha^*} = \iint_{\mathbb{R}^2} \frac{dx\wedge dy}{\pi} e^{-(b+1)(x^2+y^2)} (x-i y)^\nu (x+i y) e^{c(x-i y)}\,.
\end{equation}
Now, by introducing polar coordinates $\rho\in[0,\infty)$ and $\theta\in[0,2\pi]$ we can write
\begin{equation}
    I_\nu(b,c) = \frac{1}{\pi} \int_0^\infty d\rho \int_0^{2\pi} d\theta \, e^{i(1-\nu)\theta}\rho^{\nu+2}e^{-(b+1)\rho^2} e^{c \rho e^{-i\theta}} = \frac{1}{\pi} \int_0^\infty d\rho\, I_\nu(\rho;b,c)\, \rho^{\nu+2} e^{-(b+1)\rho^2}\,,
    \label{eq:int_nu_polar}
\end{equation}
where
\begin{equation}
    I_\nu(\rho;b,c) = \int_0^{2\pi} d\theta \, e^{i(1-\nu)\theta} e^{c \rho e^{-i\theta}} = -i\oint_{C_0(1)} d\alpha \, \alpha^{\nu-2} e^{c\rho \alpha}
\end{equation}
and $C_{\alpha_0}(R)$ denotes the circle in the complex plane with radius $R$ centered at $\alpha=\alpha_0$. Now, the integral above can be easily computed with complex contour integration techniques. In particular, by invoking Cauchy's residue theorem we can compute $I_\nu(\rho;b,c)$ as a sum of residues at its poles, the existence of which depends on the value of $\nu$. For $\nu=0$ the integrand has a pole of order $2$ in $\alpha=0$, for $\nu=1$ a simple pole, and for $\nu>1$ the integrand is an entire function and thus has no poles. Therefore we are guaranteed that $I_\nu(\rho;b,c)=0$ for all $\nu>1$. We have
\begin{equation}
    I_\nu(\rho;b,c) =\begin{cases}
         2\pi \, \underset{\alpha=0}{\text{Res}}\,\Big(\alpha^{\nu-2}e^{c\rho\alpha}\Big) & \text{if } \nu\leq 1\,, \\ 0 & \text{if } \nu>1\,.
    \end{cases}
\end{equation}
Now, the residue at $\alpha=0$ can be computed as
\begin{enumerate}
        \item $\bm{\nu=1}$\\
        In this case the pole at $\alpha_0$ is simple and therefore we have
        \begin{equation}
            \underset{\alpha=0}{\text{Res}}\,\Big(\alpha^{\nu-2}e^{c\rho\alpha}\Big) = \lim_{\alpha\to 0} e^{c\rho\alpha} = 1\,;
        \end{equation}
        \item $\bm{\nu=0}$\\
        In this case the pole at $\alpha_0$ is of order $2$ and therefore we have
        \begin{equation}
            \underset{\alpha=0}{\text{Res}}\,\Big(\alpha^{\nu-2}e^{c\rho\alpha}\Big) = \lim_{\alpha\to 0} \frac{d}{d\alpha}e^{c\rho\alpha} = c\rho \lim_{\alpha\to 0} e^{c\rho\alpha} = c\rho\,.
        \end{equation}
    \end{enumerate}
We then have
\begin{equation}
    I_\nu(\rho;b,c) = 2\pi \big( \delta_{\nu,1} + c\rho \,\delta_{\nu,0} \big)\,.
\end{equation}
By inserting this expression in Eq.\eqref{eq:int_nu_polar} we obtain
\begin{equation}
    I_\nu(b,c) = 2 \int_0^\infty d\rho\, \big( \delta_{\nu,1} + c\rho \,\delta_{\nu,0} \big)\, \rho^{\nu+2} e^{-(b+1)\rho^2} = 2 \big( \delta_{\nu,1} + c \delta_{\nu.0}\big)\int_0^\infty d\rho\, \rho^{3}e^{-(b+1)\rho^2}\,,
\end{equation}
and by using the fact that
\begin{equation}
    \int_0^\infty d\rho\, \rho^{n}e^{-s\rho^2} = \frac{1}{2} s^{-\frac{n+1}{2}}\Gamma(\frac{n+1}{2}) \quad\text{if } \Re(s)>0\,,
\end{equation}
we can finally write
\begin{equation}
    I_\nu(b,c) = \frac{\Gamma(2)}{(b+1)^2}\Big( c\,\delta_{\nu,0} + \delta_{\nu,1} \Big) =  \frac{1}{(b+1)^2}\Big( c\,\delta_{\nu,0} + \delta_{\nu,1} \Big)\,.
\end{equation}

\end{proof}

By using the lemma above, we can compute the integrals $I_\nu(\alpha)$ of
\begin{equation}
    I_\nu(\alpha)=I_\nu\Big(\frac{1}{2},-(\alpha+2\beta)\Big)= \frac{4}{9}\Big( \delta_{\nu,1}- (\alpha+2\beta)\delta_{\nu,0} \Big)\,,
\end{equation} 
and we finally arrive at the expression
\begin{equation}
    \Omega(\alpha,\theta) = \frac{4}{9}e^{-2|\beta|^2} e^{-\frac{1}{2}|\alpha|^2} \left\{ e^{-2\beta\alpha}\alpha^*\left[ \theta_1^* - (\alpha+2\beta)\theta_0^* + (\alpha+2\beta)\sum_\nu \theta_\nu\alpha^\nu \right] + h.c. \right\}\,.
\end{equation}
which we shall now use in Eqs. \eqref{eq:hamilton_eqs} to get the solution $\alpha(t),\theta_\mu(t)$. The equation of motion for $\alpha(t)$ can be now further simplified by noting that the correction $\epsilon \frac{\partial\Omega}{\partial\alpha^*}$, is a linear function of the coordinates $\theta_\mu$. This means that such correction is of second order in $\epsilon$ since $\dot{\theta}_\mu\sim o(\epsilon)$. This is just a trivial consequence of the fact that the solutions $\theta_\mu(t)$ have (as they should) vanishing zeroth order term in $\epsilon$ as such coordinates play no role for a linear oscillator. Finally the equations of motion read
\begin{equation}
\begin{split}
    &\dot{\alpha} = -i \Big( \omega_0\alpha - 4\epsilon(\alpha+\beta)e^{-4|\alpha+\beta|^2} \Big) + o(\epsilon^2)\,.\\
     &\dot{\theta}_\mu = -i\epsilon \frac{4}{9} e^{-\frac{1}{2}(|\alpha|^2+4|\beta|^2)}e^{-2\beta\alpha}\alpha^* \Big( \delta_{\mu,1}-(\alpha+2\beta)(\delta_{\mu,0} + \alpha^\mu)  \Big)\,, 
\end{split}
\label{eq:hamilton_final}
\end{equation}

Clearly, an exact solution of Eq. \eqref{eq:hamilton_final} is not feasible. Nonetheless, we can study the behaviour of these corrections in the regime of large and small amplitudes $|\alpha|$. \\

\subsection{Large amplitudes}
For $|\alpha|\gg 1$, the corrections to the the linear dynamics vanish exponentially in $|\alpha|^2$, so that the equations of motion reduce to
\begin{equation}
\begin{split}
    & \dot{\alpha} \approx -i \omega_0\alpha \,,\\ &\dot{\theta}_\mu \approx 0\,.   
\end{split}
\label{eq:hamilton_eqs_large}
\end{equation}
so that the coordinates $\theta_\mu$ become cyclic and the equation for $\alpha$ reproduces the equation of motion for a linear harmonic oscillator.\\
If we are interested in the deviations from the linear behaviour in powers of $e^{-|\alpha|^2}$, we can focus on the equation for $\alpha$ which for large amplitudes becomes
\begin{equation}
    \dot{\alpha} \approx -i \Big( \omega_0 - 4\epsilon e^{-4|\alpha|^2}   \Big)\alpha\,,
\end{equation}
for small times the evolution of a coherent state with amplitude $\alpha(t_0)=\alpha_0$ will then be
\begin{equation}
    \alpha(t_0+dt) = \alpha_0 + \dot{\alpha}(t_0)dt + o(dt^2) = \alpha_0 -i ( \omega_0 - 4\epsilon e^{-4|\alpha_0|^2} ) \alpha_0 dt + o(dt^2) \approx \alpha_0 e^{-i(\omega_0-4\epsilon e^{-4|\alpha_0|^2})dt}
\end{equation}
which gives us a new frequency
\begin{equation}
    \omega \approx \omega_0 - 4\epsilon e^{-4|\alpha_0|^2}\,,
\end{equation}
corresponding to the frequency shift
\begin{equation}
    \frac{\delta\omega_0}{\omega_0} = -\frac{4\epsilon}{\omega_0} e^{-4|\alpha_0|^2}\,.
\end{equation}
We can then conclude that for large amplitudes the frequency shift is negative, as confirmed by the numerical simulations in Fig.\ref{fig:osc_freq}. Furthermore, the frequency shift vanishes exponentially in the initial amplitude $|\alpha_0|$.

\subsection{Small amplitudes}
In order to study the opposite regime, i.e. when $|\alpha|\ll 1$, we perform the following approximation. We take a first order expansion of Eq.\eqref{eq:hamilton_final} in $\alpha$ and $\alpha^*$, allowing us to write
\begin{equation}
    \dot{\theta}_0 \approx  +i\epsilon \frac{8}{9} e^{-2|\beta|^2}\alpha^* \beta\,, \quad  \dot{\theta}_1 \approx  -i\epsilon \frac{4}{9} e^{-2|\beta|^2}\alpha^*\,, \quad\text{and}\quad  \dot{\theta}_\mu \approx 0 \quad \forall \mu\geq 2\,.
\end{equation}
The equations above have the following implications: first, all variables $\theta_\mu$ with $\mu\neq0,1$ are cyclic for small amplitudes and therefore when starting from a coherent state with $\theta_\mu= 0$, $\mu\neq 0,1$, we have $\theta_\mu(t)\equiv 0\,\forall t$. As per $\theta_0$ and $\theta_1$, the corresponding equations of motion contains $\alpha^*$ in a term that is proportional to $\epsilon$. By expanding $\alpha$ in powers of $\epsilon$ we have the expression
\begin{equation}
    \alpha(t) = \alpha_l(t) ( 1 + \epsilon \,\eta(t) ) + o(\epsilon^2)\,,
\end{equation}
where $\alpha_l(t)=\alpha_0 e^{-i\omega_0 t}$ is the solution for a linear  ($\epsilon=0$) harmonic oscillator. Therefore in order to be consistent with a first order treatment in $\epsilon$, we should write
\begin{equation}
\begin{split}
        &\dot{\theta}_0 = +i\epsilon \frac{8}{9} e^{-2|\beta|^2} \beta\alpha^*_0 e^{i\omega_0 t} + o(\epsilon^2) \quad \implies \quad\theta_0(t) = \frac{8}{9} \frac{\epsilon}{\omega_0}\beta e^{-2|\beta|^2}\alpha_0^* \big( e^{i\omega_0 t} -1 \big)\,;\\
        &\dot{\theta}_1 = -i\epsilon \frac{4}{9} e^{-2|\beta|^2}\alpha^*_0 e^{i\omega_0 t} + o(\epsilon^2) \qquad\quad\quad \implies \quad\theta_1(t) = -\frac{4}{9} \frac{\epsilon}{\omega_0} e^{-2|\beta|^2}\alpha_0^* \big( e^{i\omega_0 t} -1 \big)\,.
\end{split}
\label{eq:thetas}
\end{equation}
As per $\alpha(t)$, we go back to the corresponding equation of motion and by performing again an expansion in $\alpha,\alpha^*$ to first order we arrive at
\begin{equation}
    \dot{\alpha} = -i \Big[ \big(\omega_0 - 4\epsilon e^{-4|\beta|^2}(1-4|\beta|^2) \big)\alpha + 4\epsilon |\beta|^2 e^{-4|\beta|^2}\alpha^* - 4\epsilon \beta e^{-4|\beta|^2}  \Big]\,.
    \label{eq:alpha_small}
\end{equation}
With the definitions 
\begin{equation}
    w = 4 e^{-4|\beta|^2}(1-4|\beta|^2)\,,\quad p = 16|\beta|^2 e^{-4|\beta|^2}\,, \quad \kappa = 4 |\beta| e^{-4|\beta|^2}\,,
\end{equation}
we can write Eq.\eqref{eq:alpha_small} (together with the corresponding equation for $\alpha^*$) in the following form:
\begin{equation}
\begin{split}
    \dot{\alpha} & = -i \Big( (\omega_0-\epsilon w)\alpha + \epsilon p\alpha^* \Big) - \epsilon \kappa\,,\\ \dot{\alpha}^* & = +i \Big( (\omega_0-\epsilon w)\alpha^* + \epsilon p\alpha  \Big) - \epsilon \kappa\,.   
\end{split}
\label{eq:smallampl}
\end{equation}
The equations above have solution
\begin{equation}
    \alpha(t)=\alpha_0 e^{-i t (\omega_0-w \epsilon )}+ \frac{\epsilon}{\omega_0} \Big[ i \kappa  \big(1-\cos (t (\omega_0-w \epsilon ))\big)-(\kappa +i \alpha_0 p) \sin (t (\omega_0-w \epsilon )) \Big]
\end{equation}
which is a periodic function with frequency $\omega = \omega_0 - \epsilon w$, and a corresponding relative frequency correction
\begin{equation}
    \frac{\delta \omega_0}{\omega_0} = - \frac{\epsilon w}{\omega_0} =  4\frac{\epsilon}{\omega_0} e^{-4|\beta|^2}(4|\beta|^2-1)\,.
\end{equation}

\subsubsection{Spatial dynamics}

In order to find an expression for the new ground state and the corrections to the frequency of oscillations near the ground state, we study the spatial dynamics of the states $\ket{\alpha,\theta}$. From now on, since for small amplitude there are only two non-cyclic theta variables, we introduce the notation $\theta\equiv\{\theta_0,\theta_1\}$ and write
\begin{equation}
\begin{split}
    & \langle x(t) \rangle = \bra{\alpha(t),\theta(t)} x \ket{\alpha(t),\theta(t)} = \bra{\alpha}\big( \mathbb{1} -\epsilon \Pi(\alpha(t),\theta(t)) \big)\, x\,\big(\mathbb{1} -\epsilon \Pi(\alpha(t),\theta(t))\big)\ket{\alpha} \\ & = \bra{\alpha(t)}x\ket{\alpha(t)} + \epsilon \bra{\alpha(t)} [x,\Pi(\alpha(t),\theta(t)] \ket{\alpha(t)} + o(\epsilon^2)  \,,
\end{split}
\end{equation}
as before, since $\theta(t)$ is already of first order in $\epsilon$ we should write
\begin{equation}
    \langle x(t) \rangle = \Re(\alpha(t)) + o(\epsilon^2)\,.
\end{equation}
Surprisingly this tells us that, up to second order terms in $\epsilon$, the spatial dynamics of the oscillator is unaffected by the time evolution of $\theta$, and only depends on the dynamics in the complex plane $\alpha$. 
By using the solution found in the previous section, we can write the spatial dynamics for small amplitudes as
\begin{equation}
     \langle x(t) \rangle = \alpha_0 \cos\big((\omega_0-w \epsilon)t\big) - \frac{\epsilon k}{\omega_0} \sin\big((\omega_0-w \epsilon)t\big)
\end{equation}
We can thus compute the new equilibrium position $\alpha_v$ by setting $\dot{\alpha}=\dot{\alpha}^*=0$, which gives
\begin{equation}
    \alpha_v = \frac{i\kappa}{\omega_0 - p-w} = \frac{i\kappa}{\omega_0} + o(\epsilon^2) =i 4|\beta|\frac{\epsilon}{\omega_0}e^{-4|\beta|^2}
\end{equation}

\subsection{Ground state}
From the discussion above we can conclude that the coherent state $\alpha=0$ is not a fixed point of the dynamics anymore, and therefore the state $\ket{0}$ is not the ground state of the system. However, we have found an expression for the new coherent state amplitude corresponding to the new spatial fixed point, and the state $\ket{\alpha_v}\equiv\ket{\alpha_v,0,0}$ turns out to be the perturbed ground state up to second order terms in $\epsilon$, as its time evolution is
\begin{equation}
   \ket{\alpha_v,0} \longmapsto^t \ket{\alpha_v,\theta^v(t)}\,, \quad\text{where}\quad \theta_v(t)= \{\theta_0^v(t),\theta_1^v(t)\} \,,
\end{equation}
where, from Eq.\ref{eq:thetas} we know that
\begin{equation}
    \theta_0^v(t) = \frac{8}{9} \frac{\epsilon}{\omega_0}(2\beta-1) e^{-2|\beta|^2}\alpha_v^* \big( e^{i\omega_0 t} -1 \big)\,, \quad\text{and}\quad \theta_1^v(t)=-\frac{8}{9} \frac{\epsilon}{\omega_0} e^{-2|\beta|^2}\alpha_v^* \big( e^{i\omega_0 t} -1 \big)\,.
\end{equation}
Now, from the expressions above we note that both $\theta_0^v$ and $\theta_1^v$ are proportional to $\alpha_v^*$, which is proportional to $\epsilon$ meaning that to first order in $\epsilon$ we can write
\begin{equation}
   \ket{\alpha_v,0} \longmapsto^t \ket{\alpha_v,0)}\,.
\end{equation}

This means that the frequency of spatial oscillations of a coherent state $\ket{\alpha_0}$ for small amplitudes will be $\omega_0 - w$, regardless of $\theta(t)$, and this result reproduces exactly the shifts observed in numerical simulations for small displacements of the true ground state (see Fig.\ref{fig:osc_freq}).

\end{document}